\documentclass[acmsmall]{acmart}
\AtBeginDocument{%
  }

\usepackage{multirow} 
\usepackage{subfigure}
\usepackage{url}
\usepackage{booktabs}
\usepackage[ruled,vlined,linesnumbered]{algorithm2e}
\usepackage{longtable}

\begin{document}

\title{Enhanced Influence-aware Group Recommendation for Online Media Propagation}

\author{Chengkun He}
\email{ck131102@hotmail.com}
\orcid{0000-0003-3144-7166}
\affiliation{%
  \institution{RMIT University}
  \city{Melbourne}
  \state{VIC}
  \country{Australia}
}

\author{Xiangmin Zhou}
\orcid{0000-0002-1302-818X}
\affiliation{%
  \institution{RMIT University}
  \city{Melbourne}
  \state{VIC}
  \country{Australia}
  }
\email{xiangmin.zhou@rmit.edu.au}

\author{Chen Wang}
\affiliation{%
  \institution{Data61 CSIRO}
  \city{Sydney}
  \country{Australia}
}
\email{chen.wang@data61.csiro.au}

\author{Longbing Cao}
\affiliation{%
 \institution{Macquarie University}
 \city{Sydney}
 \state{NSW}
 \country{Australia}
 }
\email{longbing.cao@mq.edu.au}

\author{Jie Shao}
\affiliation{%
  \institution{University of Electronic Science and Technology of China}
  \city{Chengdu}
  \state{SiChuan}
  \country{China}}
\email{shaojie@uestc.edu.cn}

\author{Xiaodong Li}
\affiliation{%
  \institution{RMIT University}
  \city{Melbourne}
  \state{VIC}
  \country{Australia}
  }
\email{xiaodong.li@rmit.edu.au}

\author{Guang Xu}
\affiliation{%
  \institution{New Aim Pty Ltd}
  \city{Melbourne}
  \state{VIC}
  \country{Australia}
  }
\email{guang.xu@newaim.com.au}

\author{Carrie Jinqiu Hu}
\affiliation{%
  \institution{New Aim Pty Ltd}
  \city{Melbourne}
  \state{VIC}
  \country{Australia}
  }
\email{carrie.hu@newaim.com.au}

\author{Zahir Tari}
\affiliation{%
  \institution{RMIT University}
  \city{Melbourne}
  \state{VIC}
  \country{Australia}
  }
\email{zahir.tari@rmit.edu.au}

\renewcommand{\shortauthors}{Chengkun et al.}
\newcommand{\eat}[1]{}
\newcommand{\hck}[1]{\textcolor{blue}{[#1]}}
\newtheorem{myDef}{Definition}

\begin{abstract}
Group recommendation over social media streams has attracted significant attention due to its wide applications in domains such as e-commerce, entertainment, and online news broadcasting. By leveraging social connections and group behaviours, group recommendation (GR) aims to provide more accurate and engaging content to a set of users rather than individuals. Recently, influence-aware GR has emerged as a promising direction, as it considers the impact of social influence on group decision-making. In earlier work, we proposed Influence-aware Group Recommendation (IGR) to solve this task. However, this task remains challenging due to three key factors: the large and ever-growing scale of social graphs, the inherently dynamic nature of influence propagation within user groups, and the high computational overhead of real-time group-item matching. 

To tackle these issues, we propose an Enhanced Influence-aware Group Recommendation (EIGR) framework. First, we introduce a Graph Extraction-based Sampling (GES) strategy to minimise redundancy across multiple temporal social graphs and effectively capture the evolving dynamics of both groups and items. Second, we design a novel DYnamic Independent Cascade (DYIC) model to predict how influence propagates over time across social items and user groups. Finally, we develop a two-level hash-based User Group Index (UG-Index) to efficiently organise user groups and enable real-time recommendation generation. Extensive experiments on real-world datasets demonstrate that our proposed framework, EIGR, consistently outperforms state-of-the-art baselines in both effectiveness and efficiency.
\end{abstract}




\received{01 July 2025}

\maketitle

\section{Introduction}
The rapid proliferation of online platforms has led to a dramatic increase in the volume of social media streams, particularly with the growing prevalence of e-commerce applications. These social media streams often carry critical content such as digital advertisements and event notifications, which are intended to reach a broad audience either directly or through social propagation. The dissemination of such content is significantly influenced by recommendation systems and the social influence of the users receiving these recommendations. This has brought growing attention to the development of influence-aware recommendation techniques.
Influence-aware recommendation plays a crucial role in a variety of applications, including online product promotion and real-time news delivery. For example, an e-commerce platform may distribute digital advertisements to users who are likely not only to make purchases themselves but also to influence their social contacts to engage with the content. In real-world scenarios, social media users are often organised into sub-communities or user groups, which interact with each other through social relationships such as friendship or shared interests. These dynamics highlight the need for influence-aware group recommendation systems that can operate effectively over continuous social media streams.

We study the continuous influence-aware group recommendation over social communities. Given an incoming social item $v$ and historical user groups $\{g_i\}$, we aim to automatically learn an item embedding $e_v$ and the embeddings of user groups $\{e_{gi}\}$, predict the dynamic interests and influences of user groups with respect to $v$, and return a list of user groups $\{g_i\}$ that have the highest probability scores to interact with $v$. For influence-aware group recommendation, three key issues need to be addressed. First, a novel data model is required to well capture the dynamic attributes (e.g., item categories, location, rating, popularity, etc) of items and user groups, as well as the dynamic interactions between incoming items and user groups for effective data representation and group interest prediction. As incoming items are new to a social community, the number of user interactions over them may be small when they are just uploaded into social networks. With the item propagation over social networks, the interactions between the item and its user groups increase, which further causes the changes of item attributes. For example, a product promoted by an influencer may become very popular overnight during its propagation, which leads to the change of its popularity attribute.
A recommender system should be able to handle the interaction sparsity of items and capture the temporal dynamics of item-group interactions and item attributes for high quality recommendation. Second, a novel model is required to well capture both the impact of incoming items on user groups and the dynamics of group influence for effective group influence prediction. While different users may have different influences on the propagation of an item, a user may have different influences on various items. Third, due to the dynamics of user activities, a social media broadcaster may be active in the morning while there may be more online users in the evening. Such user activeness dynamics affect the propagation of online items over different time periods. Thus, it is inappropriate to treat user groups equally and statically for all incoming items. A good influence prediction model should reflect this influence dynamics with respect to the incoming items over time for more accurate influence prediction. 
Finally, efficient indexes are required to organize and search the user group database, reducing the cost of matching the streaming items with social user groups. The number of user groups in social networks is big, while the presentations of groups and items are complex. According to the statistics in 2022, Yelp has more than 178 million unique visitors monthly\footnote{https://review42.com/resources/yelp-statistics/}, forming a big number of user groups by online interactions.
It is essential to avoid unnecessary group-item matching for real time recommendation.
Existing group recommendations are score aggregation-based \cite{DBLP:conf/recsys/BaltrunasMR10} and preference aggregation-based \cite{DBLP:conf/kdd/YuanCL14,DBLP:conf/icde/YinW0LYZ19,DBLP:conf/icde/ZhangGJ021}. However, score aggregation-based methods are inflexible, while preference aggregation-based methods treat all users equally and do not consider the influence difference of users. Recent attention-based approaches~\cite{DBLP:conf/sigir/Cao0MAYH18,DBLP:conf/icde/GuoYW0HC20} take the influence of users as their weights in group preference aggregation. However, they only consider user influence in a static manner, which is inapplicable to applications with dynamic group influence changes. Turning to social streaming, neither existing memory-based  \cite{DBLP:conf/sigmod/HuangCZJX15,DBLP:conf/cikm/SubbianAH16} nor model-based recommendation \cite{DBLP:conf/icde/ZhangLXLY0XCM23,DBLP:conf/icde/ZhouQ0CZ19,DBLP:conf/kdd/ZhangDXDDW22} considers the dynamic influence of groups in recommendation generation. 

Due to the limitations of existing approaches, we propose influence-aware group recommendation (IGR) for social media propagation \cite{DBLP:conf/icdm/HeZ0C0T24}. We first design a Group Graph Convolutional Network (GroupGCN) to learn item and group embeddings by capturing group-level relationships. To model the temporal evolution of group interests, we extend GroupGCN into a sequential architecture named Temporal GroupGCN-RNN-Autoencoder (TGGCN-RA), which integrates recurrent structures and autoencoding mechanisms for effective sequence modelling. Next, we construct a Group Relationship Graph (GREG) to represent inter-group social connections. Based on GREG, we adopt the Independent Cascade (IC) model to simulate influence propagation across user groups. Finally, we generate real-time recommendations by measuring the relevance between incoming items and user groups, and employ optimisation strategies to enhance the efficiency of this process.

As a second step, extending our IGR proposed in \cite{DBLP:conf/icdm/HeZ0C0T24}, we further optimise the model training for faster convergence and better generalisation, introduce novel influence estimation to capture the dynamics of influence propagation, and design a lightweight retrieval mechanism to accelerate the matching process. Specifically, we first develop a sampling-based algorithm, GES, which preserves the original data distribution and effectively captures groups exhibiting interest drift, while minimising redundancy across multiple graphs of GREG. Based on GREG, we then propose a Dynamic Item-aware Information PROpagation Graph (DI²PROG) model to capture the evolving nature of group influence. Leveraging DI²PROG, we introduce the DYnamic Independent Cascade (DYIC) model to simulate influence-based media propagation across groups. Finally, we perform real-time recommendation over social media streams by evaluating the relevance between each incoming item and user group. To support this process efficiently, we design a novel hash-based indexing scheme, UG-Index, which significantly accelerates group matching and recommendation generation.
\eat{Our main contributions are summarized as follows. 
\begin{itemize}
    \item We propose a novel GroupGCN model and its extension  TGGCN-RA to temporal sequences. GroupGCN is robust to the data sparsity and can handle the media dynamics for effective data presentation, while TGGCN-RA enables the accurate group interest prediction for next time point.
    \item We propose a novel DYIC model that well captures the group activeness, group similarity, and propagation willingness to items. It simulates the information propagation over a new dynamic item-aware graph DI$^2$PROG.
    \item We propose a novel GES algorithm that samples the edges of GREG. GES keeps the distribution of the sampled dataset and captures the interest drift of the groups, enabling effective and efficient TGGCN-RA training. 
    \item We design a two-level hash-based index. A set of bidirectionally linked blocks keep the ordered Z-order values generated by locality sensitive hashing (LSH) and group features. A chained hash table keeps the Z-order value positions.
\end{itemize}
}
The early version of this work has been published in \cite{DBLP:conf/icdm/HeZ0C0T24}. Compared to that work, we have made several new
contributions:
\begin{itemize}
    \item We propose a novel GES algorithm that samples the edges of GREG. GES keeps the distribution of the sampled dataset and captures the interest drift of the groups, enabling effective and efficient TGGCN-RA training. 
    \item We propose a novel DYIC model that well captures the group activeness, group similarity, and propagation willingness to items. It simulates the information propagation over a new dynamic item-aware graph DI$^2$PROG.
    \item We design a two-level hash-based index. A set of bidirectionally linked blocks keep the ordered Z-order values generated by locality sensitive hashing (LSH) and group features. A chained hash table keeps the Z-order value positions.
\end{itemize}

The remainder of this paper is organised as follows. Section~\ref{sec:rw} reviews the related work on stream recommendation and group recommendation. Section \ref{sec:mth} introduces our enhanced influence-aware group recommendation framework. Section \ref{sec:exp} presents the experimental evaluation and analysis. Finally, Section \ref{sec:conclusion} concludes the paper and outlines future directions.

\section{Related Work}
\label{sec:rw}
\subsection{Stream Recommendation}
Traditional stream recommendation methods \cite{DBLP:conf/sigmod/HuangCZJX15,DBLP:conf/cikm/SubbianAH16} typically store data in memory to enable real-time recommendation. For instance, TencentRec \cite{DBLP:conf/sigmod/HuangCZJX15} adopts a practical item-based collaborative filtering approach, featuring scalable incremental updates and real-time pruning. However, it suffers from high computational overhead when dealing with large-scale data. To improve efficiency, Subbian et al. \cite{DBLP:conf/cikm/SubbianAH16} employ min-hash to approximate item similarities. Their hash-based data structure enables efficient representation of new rating records, allowing the model to be updated in real time while capturing short-term interest drift. Nevertheless, this method overlooks users’ long-term preferences.
To address these limitations, model-based approaches have been proposed \cite{DBLP:conf/icde/ZhouQ0CZ19,DBLP:conf/kdd/ZhangDXDDW22,DBLP:conf/icde/ZhangLXLY0XCM23}. BiHMM \cite{DBLP:conf/icde/ZhouQ0CZ19} captures both long-term and short-term user interests through probabilistic entity matching. NDB \cite{DBLP:conf/kdd/ZhangDXDDW22} models user attention mechanisms using RNNs and learns a randomly weighted neural network to predict user-item relevance. eLiveRec \cite{DBLP:conf/icde/ZhangLXLY0XCM23} employs disentangled encoders to represent users’ cross-domain and domain-specific intentions, and leverages multi-task learning to capture both intra-channel and inter-channel behaviours.
However, despite their effectiveness, these methods largely overlook the role of user influence in shaping preferences and propagating content, which is critical in social media environments.
%in recommendation

\subsection{Group Recommendation}
Existing group recommendation approaches primarily focus on group aggregation, which can be broadly categorised into score aggregation \cite{DBLP:series/sci/BorattoC11,DBLP:conf/recsys/BaltrunasMR10,DBLP:journals/pvldb/Amer-YahiaRCDY09,DBLP:journals/tkde/QinZCHZ20} and preference aggregation \cite{DBLP:conf/kdd/YuanCL14,DBLP:conf/sigir/Cao0MAYH18,DBLP:conf/icde/ZhangGJ021,DBLP:journals/tois/GuoYCZZ22,DBLP:conf/www/WuX0J0ZY23} methods.
Score aggregation combines the individual recommendation lists of group members into a unified group recommendation. Common strategies include average satisfaction \cite{DBLP:series/sci/BorattoC11}, least misery \cite{DBLP:journals/pvldb/Amer-YahiaRCDY09}, and maximum pleasure \cite{DBLP:conf/recsys/BaltrunasMR10}. However, these strategies are heuristic and static, failing to adapt to the evolving structure and shifting interests within user groups.

Preference aggregation, in contrast, models the collective preferences of group members by integrating their individual interests. For example, COM \cite{DBLP:conf/kdd/YuanCL14} builds group profiles based on user preferences and their relative influence. AGREE \cite{DBLP:conf/sigir/Cao0MAYH18} combines attention mechanisms with neural collaborative filtering to learn user-specific weights over items. Zhang et al. \cite{DBLP:conf/icde/ZhangGJ021} construct a heterogeneous graph with initiators, users, and items as nodes, and interactions such as reviews and purchases as edges; a GCN-based model is then used to learn user and item representations.
MGBR \cite{DBLP:conf/icde/ZhaiLYX23} addresses group-buying by decomposing it into two sub-tasks: multi-view embedding learning via GCNs and objective prediction via multi-task learning and MLPs. HyperGroup \cite{DBLP:journals/tois/GuoYCZZ22} builds a hypergraph where users are nodes and groups are hyperedges, learning group representations through hyperedge embeddings. ConsRec \cite{DBLP:conf/www/WuX0J0ZY23} constructs multi-view graphs and uses GNNs to encode group consensus, which are then fused into group-level preferences.
Despite their effectiveness, these preference aggregation-based methods neglect inter-group interactions and fail to model group-level influence, limiting their ability to support information propagation in social communities.
To address this gap, this work targets group recommendation for social media propagation, aiming to dynamically model item-driven group interests and group influence, so that recommended items can be effectively disseminated across social platforms. In addition, we propose optimisation strategies to support efficient model training and real-time recommendation. 
%The main notations used in this paper are summarised in Table~\ref{fbl-notations}.

%explicitly models the group relationships and captures the item-driven group interests and influence change over time. 

\section{Framework of our solution}
\label{sec:framework}

In broadcasting, the influence can be estimated by the information propagation. 
The more information a node can propagate, the more influence it has.
Especially, we focus on the information which is not only propagated but also accepted.
To model the information propagation, we take the preference of groups and the characteristics of items into consideration. 
\begin{myDef}
\textbf{Information propagation} We define the information propagation as the situation where group $g_i$ propagates information of item $v$ to group $g_j$, and group $g_j$ \textbf{accepts} the information. And the probability of an information propagation is defined as $p_{ij}^v$.
\end{myDef}

\begin{myDef}
Given social network of group $S$ and corresponding attribute features $X$ and $Y$ of groups and items $v$, our model is to predict the relevance score which denotes the preference of group to item and the corresponding influence of group.
\begin{equation}
    \left[\hat{R},Inf\right] = f_v(S,X,Y).
\end{equation}
\end{myDef}
In this work, we recommend each incoming new item to user groups in social networks, so that the influence of the recommendation can be maximized and the recommended groups like the incoming item the most. The main notations used in this paper are listed in Table \ref{tab:notation-full}.

\begin{table}[ht]
\centering
\caption{Notation Table.}
\begin{tabular}{ll|ll}
\hline
\textbf{Notation} & \textbf{Definition} & \textbf{Notation} & \textbf{Definition} \\
\hline
$a$, $b$ & Random integers for hash & $A_i$ & Activeness of group $g_i$ \\
$b^l_e$ & Edge score in layer $l$ & $B$ & Block size in UG-Index \\
$C_i$ & Node cluster in GREG & $d$ & Embedding dimension \\
$e_g$, $e_{g_i}$ & Embedding of group $g$ & $e_v$ & Embedding of item $v$ \\
$E_{ij}$ & Edges between clusters $C_i$ and $C_j$ & $F_i(a)$ & FastMap projection coordinate \\
$G_t$ & Sampled subgraph at time $t$ & $G_{\text{new}}$ & Groups recently interacted with $v$ \\
$\mathcal{G}_s$ & Group relationship graph (GREG) & $g$, $g_i$, $g_j$ & User group \\
$H(k)$ & Hash function on Z-order $k$ & $H^l$ & Weight matrix for item at layer $l$ \\
$I_u|v$ & Indicator: 1 if $u$ sampled given $v$ & $k$, $Z_i$ & Z-order value \\
$\lambda$ & Regularisation parameter & $\mathcal{L}$ & BPR loss function \\
$N_g$ & Neighbours of group $g$ & $p$ & Prime number for hashing \\
$p_e$ & Sampling probability of edge $e$ & $p_{uv}$ & Sampling prob. for edge $(u,v)$ \\
$p^v_{ij}$ & Propagation prob. from $g_i$ to $g_j$ & $p_v$ & Sampling prob. for node $v$ \\
$p_g$, $q_v$ & Latent vectors of group/item & $pos$ & Position of key in hash blocks \\
$r_{g,v}$ & Relevance score of $g$ to $v$ & $R_i$ & Group-item interactions at $t_i$ \\
$sim(g_i, g_j)$ & Similarity between groups & $\sigma(\cdot)$ & Activation function \\
$s_{ij}$ & Sampled edge set between $C_i$ and $C_j$ & $\theta$, $\theta^l_v$ & Unbiased estimator \\
$\tau$ & Propagation threshold & $T$ & Number of hash buckets \\
$\Theta_1$ & Trainable parameters & $V_g$ & Items historically interacted by $g$ \\
$v$ & A social item & $W^l$ & Weight matrix for group at layer $l$ \\
$W^v_i$ & Willingness of group $g_i$ to $v$ & $x_g$, $y_v$ & Attribute features of group/item \\
$\alpha_r$ & Preference-influence trade-off & $\alpha_v$ & Dynamics factor for $e_v$ update \\
$\beta^l$ & Bias for item encoder & $\delta_o$ & Overlap degree in GES \\
$\gamma_1$, $\gamma_2$ & Influence trade-off weights & $\hat{A}_{uv}$ & Laplacian norm coefficient \\
\hline
\end{tabular}
\label{tab:notation-full}
\end{table}

\section{The EIGR Model}
\label{sec:mth}

We propose an enhanced influence-aware group recommendation framework (EIGR) as shown in Fig. \ref{fig:framework}. The EIGR framework includes three main components: the data representation, the group influence prediction, and the group recommendation generation. 
In data representation, we leverage IGR \cite{DBLP:conf/icdm/HeZ0C0T24} to represent items and user groups, and predict the current interests of groups. To accelerate the training process, we propose a novel sampling-based strategy.
The representations of a given item and group are transferred to other components for group influence prediction and group recommendation generation. 
In group influence prediction, a group interaction graph is first constructed based on the friendship between different users in two investigated groups and the common users they share. Then, the influence of each given group and its propagation with respect to an incoming item are estimated with the support of the dynamically updated group interaction graph. 
In group recommendation generation, the group-item matching provides an influence-aware matrix factorization-based ranking function between a stream item and a user group based on the predicted group interests. A UG-Index filtering is applied to exclude the irrelevant groups. %In addition to the advanced models, we design a Hash-based UG-Index structure to optimise the efficiency. 

\begin{figure}
    \centering
    \includegraphics[width=0.7\textwidth]{ 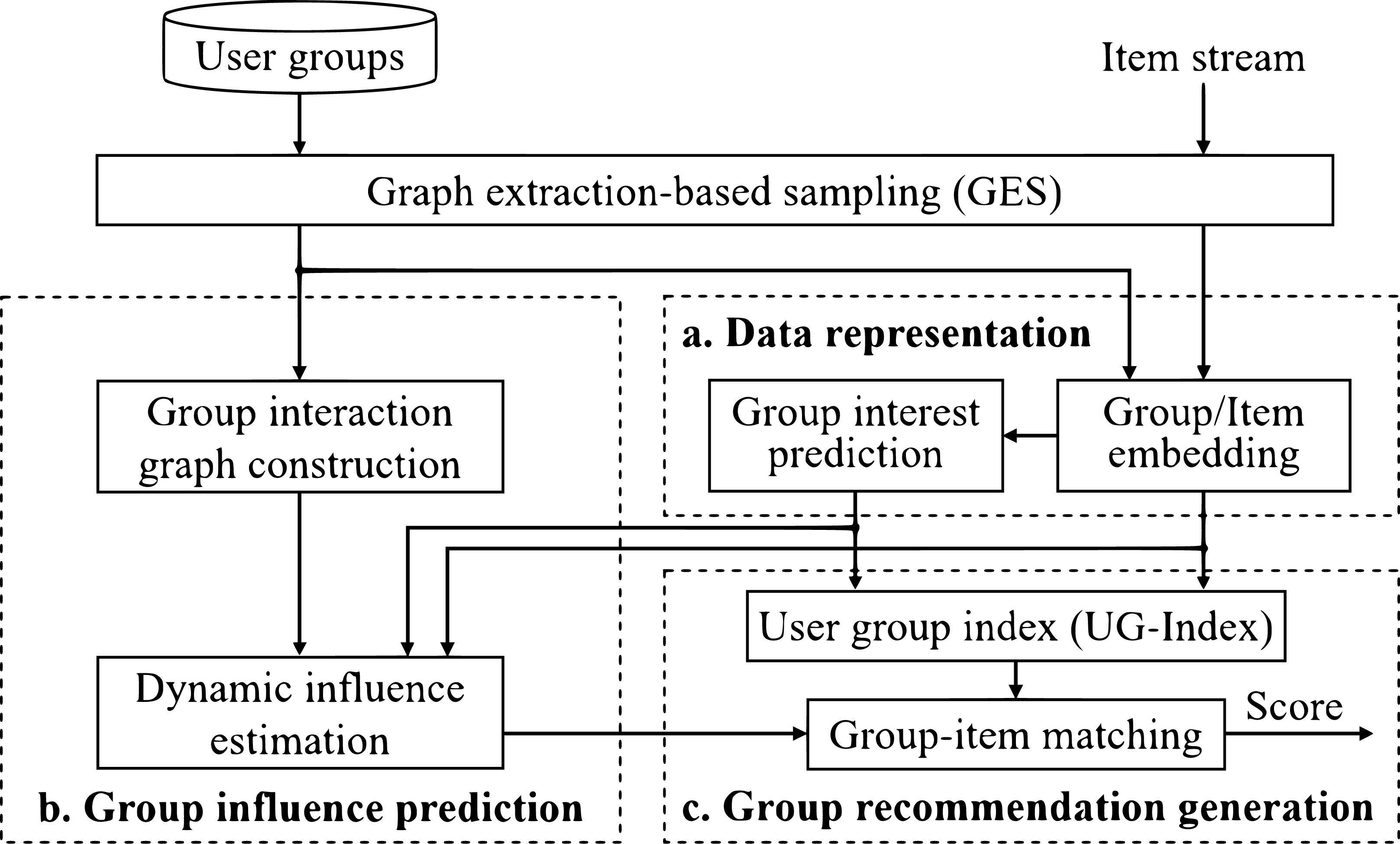}\vspace{-0ex}
    \caption{EIGR framework.}
    \Description{Diagram of the EIGR framework.}
    \label{fig:framework}\vspace{-0ex}
\end{figure}

\subsection{Data Representation}
New social users may join a group and existing users may leave anytime. The interactions between a group and an incoming item happen frequently. Thus, user interests may change over time. An incoming item has sparse interactions with users. Due to its interactions with users during its propagation, the item attributes (e.g., popularity) could change. 
Thus, our solution \cite{DBLP:conf/icdm/HeZ0C0T24} leverages the influence-aware GroupGCN to represent groups and items and extends it to TGGCN-RA to predict the group interests. We first briefly review the process of data modelling.

GroupGCN comprises three core components: the initialization layer, embedding propagation layer, and relevance prediction layer. Given a group-item pair $(g, v)$, the initialization layer first constructs latent vectors $p_g, q_v \in \mathbb{R}^d$ through hypergraph-based aggregation over group members' preferences. Attribute feature vectors $x_g, y_v \in \mathbb{R}^{\bar{d}}$ are extracted by applying Word2Vec to group/item tags and averaging the resulting tag embeddings. The initial embeddings are obtained via:
\begin{equation}
    e^0_g = \sigma(W^0\cdot [p_g,x_g]+b^0), \quad e^0_v = \sigma(H^0\cdot [q_v,y_v]+\beta^0),
\end{equation}
where $[\cdot]$ denotes concatenation, $W^0$, $H^0$ are transformation matrices, $b^0$, $\beta^0$ are biases, and $\sigma(\cdot)$ is a non-linear activation function.
To refine group embeddings, the embedding propagation layer applies graph convolution over the group relationship graph $\mathcal{G}_s$, capturing high-order structural dependencies. The neighbour aggregation is defined as:
\begin{equation}
agg_g^{l-1} = \sum_{g'\in N_g} {e^{l-1}_{g'}}/{\sqrt{\left|N_g\right|\left|N_{g'}\right|}},
\end{equation}
and the embedding update at each layer is given by:
\begin{equation}
    e^l_g = \sigma(W^l\cdot [e^{l-1}_g,agg^{l-1}_g]),\quad l\in[1,L].
\end{equation}
The final group embedding is $e_g = e^L_g$. For an item $v$, its final embedding is updated by integrating new group interactions:
\begin{equation}
\label{eq:ev}
    e_v = \alpha_v e^0_v + \frac{1-\alpha_v}{\left|G_{n\!e\!w}\right|}\sum_{g\in G_{n\!e\!w}}e_g,
\end{equation}
where $G_{\text{new}}$ denotes the set of groups recently interacting with $v$, and $\alpha_v \in [0,1]$ is a dynamics trade-off coefficient.
The relevance prediction layer computes the interaction score between a group $g$ and an item $v$ as:
\begin{equation}
\label{eq:rel_pred}
    \hat{r}_{g,v} = (1-\alpha_r) e^T_v(e_g+\left|V_g\right|^{-\frac{1}{2}}\sum_{v_i\in V_g}{e_{v_i}}) + \alpha_r\hat{s}^v_g,
\end{equation}
where $V_g$ is the set of items historically interacted with by $g$, $\hat{s}^v_g$ is the predicted item-aware group influence, and $\alpha_r$ controls the balance between preference and influence.
GroupGCN is trained using the pairwise Bayesian Personalized Ranking (BPR) loss:
\begin{equation}
    L = -\sum_{g\in G} \sum_{(u,v)\in O_g}\ln{\sigma(\hat{r}_{g,u}-\hat{r}_{g,v})} + \lambda_1\Vert\Theta_1\Vert^2_2,
\label{eq:Loss_1}
\end{equation}
where $\mathcal{O}_g = {(u, v) \mid u \in V_g, v \in \widetilde{V}_g}$ is the group-specific pairwise training set with negative samples $\widetilde{V}_g$, $\Theta_1$ denotes trainable parameters, and $\lambda_1$ is a regularisation coefficient. The model is optimised using the Adam algorithm.

To model the temporal dynamics of group interests, TGGCN-RA constructs a sequence of training quadruples $\langle R_i, V_i, G, G_s \rangle$, each containing group-item interactions up to a specific time point. A series of GroupGCNs (GGCNs) are trained on these time-based subsets to capture long-term group interests. To model short-term preferences, each GGCN is further fine-tuned using only the newly observed interactions. At each time point, the group profile is formed by concatenating the corresponding long-term and short-term embeddings.
The resulting sequence of group profiles is then fed into a recurrent neural network (RNN) to capture the temporal patterns of interest evolution. At each step, the RNN updates its hidden state based on the current group profile and past states. To predict future group interests, TGGCN-RA employs an RNN-based autoencoder, where the encoder models historical interest evolution and the decoder generates the next-step embedding.
The model is trained using a mean squared error loss between the predicted and actual group embeddings, with regularisation applied to avoid overfitting. The predicted embedding serves as the basis for estimating group influence and generating recommendations.

\subsection{Optimised TGGCN-RA Training.}
Naive way directly applies the pair-wise mini-batch training strategy over all the interaction records, each is described as a quadruple. However, training a GGCN incurs a high time cost for large graphs that have a huge number of interaction records. In addition, during the training, the convolutional operations are performed over the whole graph, and the number of involved nodes in each convolutional operation could increase exponentially as the TGGCN-RA depth grows. Thus, directly training the model over the whole batch becomes costly in practice. 

To reduce the training cost, distributed GNN processing may be a remedy \cite{DBLP:journals/pvldb/MinWHHXECH21,DBLP:journals/pvldb/ZhuZYLZALZ19,DBLP:conf/sc/TripathyYB20,DBLP:journals/pvldb/PengCSSCC22}. In these models, graph nodes are stored on different devices and training is conducted in parallel. Specifically, the training process first learns information like embeddings, gradients, or model parameters over each device, and then schedules the data communication among devices. Generally, distributed GNN models can be categorized into centralized and decentralized. In centralized models \cite{DBLP:journals/pvldb/MinWHHXECH21,DBLP:journals/pvldb/ZhuZYLZALZ19}, devices periodically send data updates to the central server. For each training iteration, the server updates the model after collecting data from all devices. However, the heavy preprocessing and complex workflow incur high costs. Moreover, variations in computation costs across devices due to data communication can hinder the speed of parallelization. Decentralized models CAGNET \cite{DBLP:conf/sc/TripathyYB20} and Sancus \cite{DBLP:journals/pvldb/PengCSSCC22} overcome the heavy preprocessing \cite{DBLP:conf/sc/TripathyYB20} and low parallelization effectiveness \cite{DBLP:journals/pvldb/PengCSSCC22}. Each of their devices generates parameters and embeddings, exchanging them directly with other devices. Sancus \cite{DBLP:journals/pvldb/PengCSSCC22} adopts a broadcast skipping to reduce the communication and training cost. All these methods assume that the graph has no redundancy. However, TGGCN-RA is trained over multiple temporally related graphs that share common or similar nodes and edges. A direct adaptation of distributed GCN models cannot alleviate the redundancy across different graphs. Given a graph at time $t-1$, some groups may not receive or only have a few new interactions when the graph evolves until the next time point $t$. Thus, the corresponding nodes and edges in the GREG remain unchanged or undergo minimal changes, leading to redundant training. This redundancy across multiple graphs cannot be handled by distributed GCN models.

Another line of research for improving training efficiency is sampling-based \cite{chen2017stochastic,chen2018fastgcn,zeng2019graphsaint,chiang2019cluster}, including layer sampling and subgraph sampling. The layer sampling-based approaches \cite{chen2017stochastic,chen2018fastgcn} build a GCN over the whole graph, sample nodes or edges from the graph to form the mini-batches, and train the model over mini-batches. On the other hand, the subgraph sampling-based methods \cite{zeng2019graphsaint,chiang2019cluster} construct subgraphs by sampling nodes and edges from the whole graph, form the mini-batches from subgraphs, and train the model using these mini-batches. Similar to distributed GNN models, these methods also assume that the graph contains no redundancy.
In addition, they cannot capture the temporal relationship between multiple graphs, and their predefined constraints or sampling strategies are very complex, which is not suitable for streaming data.

To overcome the problems of existing strategies, we need to minimize the redundancy across multiple graphs and capture the dynamics of groups and items over the timeline in an efficient way. We have two challenges. First, the distribution of the sampled dataset may be different from that of the original one due to the removal of edges in GREG during sampling. Thus, the effectiveness of the model may be downgraded due to this data distribution change. Second, with sampling, the nodes of the original GREG could be removed, which removes all the interactions of these unselected nodes at different time points, leading to the loss of groups with interest drift. This affects the model's ability to capture the dynamics of groups. We propose a novel Graph Extraction-based Sampling (GES) strategy that maintains the distribution of the sampled dataset and well captures the groups with interest drift.

\begin{algorithm}[t]
\caption{GES: Group-aware Edge Sampling}
\label{alg:smpAlg}
\KwIn{Interaction sets $\{R_t\}$; GREG $G_s$; overlap degree $\delta_o$}
\KwOut{Subgraphs $\{G_t\}$}

$\{C_i\} \leftarrow \text{ConstructCluster}(G_s)$\;

$\{E_{ij}\} \leftarrow \text{ConstructEdgeSet}(G_s, \{C_i\})$\;

Initialize $\{num_t\}$\;

\For{$t \in [1, T]$}{
    \ForEach{$E_{ij}$ at $t$}{
        \eIf{$t = 1$}{
            $s_{ij} \leftarrow \text{SampleEdge}(E_{ij})$\;
        }{
            $\hat{s}_{ij} \leftarrow \text{SampleOverlapEdge}(E_{ij}, G_{t-1}, R_{t-1}, R_t, \delta_o)$\;
            $s_{ij} \leftarrow \text{SampleEdge}(E_{ij}, G_{t-1}, \delta_o)$\;
        }
    }
    $S_t \leftarrow s_{11} \cup \cdots \cup s_{CC} \cup \hat{s}_{11} \cup \cdots \cup \hat{s}_{CC}$\;
    $G_t \leftarrow \text{GenerateSubgraph}(G_s, S_t)$\;
}

\Return $\{G_t\}$\;
\end{algorithm}

Given interaction sets $\{R_t\}$ and a GREG $G_s$, GES generates a subgraph at each time point. Alg. \ref{alg:smpAlg} shows GES performed in two steps. The first step is sampling preparation (lines 1-3). We first cluster the graph nodes into $K_c$ groups using K-Means++. The nodes in different groups are sampled in the same proportion, which maintains the same data distribution (line 1). For any two clusters $C_i$ and $C_j$ ($i, j \leq K_c$), we keep all the edges between them in an edge set $E_{ij}$ (line 2). As such, a set of edge sets are formed for all the node groups. We initialize the number of samples for each cluster (line 3). The second step is sampling over each edge set (lines 4-12). If an edge set is to the first time point ($t=1$), we sample its edges based on the degrees of nodes linking each edge (lines 4-7). Otherwise, we first sample $\sigma_o*num_t$ overlapping edges from two edge sets at two adjacent time points (line 9). Here, two edges from two edge sets sharing common nodes at both sides of the edges are overlapping. If the interactions of two overlapping edges have changed, $\sigma_o*num_t$ edges with the highest degrees are extracted, where $\sigma_o$ is the proportion of overlapping edges. Then, $(1-\sigma_o)*num_t$ edges are sampled as performed for $t=1$ (line 10). Finally, the sampled edges are used to construct the subgraph $G_t$ for training the model at the time point $t$ (lines 11-12). Given a GREG $G_s$ with $N$ nodes and $N_e$ edges and the number of clusters $K_c$, the time cost of GES is $\mathcal{O}(N*K_c+N_e)$.
With GES, we achieve both node and edge level unbias, and ensure that each subgraph captures groups with interest drift.
To ensure that edge samplers are unbiased and information loss is minimized, we design unbiased estimators and minimize the variance. Given a node $v$, we form the propagation at layer $l$:
\begin{equation}
    x^{l}_v=\sigma(\sum_{u\in N_v}\hat{L}_{uv}W^{(l-1)}x^{(l-1)}_u),
\end{equation}
where $\hat{L}_{uv}={1}/{\sqrt{\left|N_v\right|\left|N_{u}\right|}} $ is graph Laplacian norm and $\sigma$ is activation function.
Let $\bar\theta^{(l-1)}_v = \sum_{u\in N_v}\hat{L}_{uv}W^{(l-1)}x_u^{(l-1)}$, we design unbiased estimator $\theta_v^{(l-1)}$, holding the condition: $E(\theta_v^{l}) = \bar\theta^{(l-1)}_v$.
Here, $E(\theta_v^{l})$ is the expectation of $\theta_v^{l}$. The unbiased estimator for a node $v$ is computed as:
\begin{equation}
    \theta_v^{l} = \sum_{u\in N_v}(\hat{L}_{uv}/\alpha_{uv})\hat{x}^{(l-1)}_u\mathbb{I}_{u|v},
\label{eq:est_node}
\end{equation}
where $\hat{x}^{(l-1)}_u=W^{(l-1)}x^{(l-1)}_u$, $\alpha_{uv}={p_{uv}}/{p_v}$ is the aggregator normalization to guarantee the unbiased estimator, $p_{u,v}$ is the probability of an edge $(u, v)$ being sampled in a subgraph, $p_{v}$ is the probability of a node $v$ being sampled,
and $\mathbb{I}_{u|v}\in \{0,1\}$ is the indicator function (when $u$ is sampled, $\mathbb{I}_{u|v}=1$; otherwise, $\mathbb{I}_{u|v}=0$). Then we can define the unbiased estimator for one subgraph $G_s$ as below:
\begin{equation}
    {\theta}=\sum_l\sum_{v\in G_s}(\theta_v^{l}/p_v)=\sum_l\sum_{e\in G_s}(b^{l}_{e}/p_e)\mathbb{I}_{e}^{l},
\label{eq:est_graph}
\end{equation}
where $p_{e}$ is the probability of an edge $e$ being sampled in a subgraph, $\mathbb{I}_{e}^{l}=1$ if the edge $e$ is sampled; otherwise, $\mathbb{I}_{e}^{l}=0$, and $b^{l}_e=\hat{L}_{uv}(\hat{x}^{(l-1)}_u+\hat{x}^{(l-1)}_v)$. %Next, we prove both estimators $\theta_v^{l}$ and $\theta$ are unbiased.
\begin{theorem}
   $\theta_v^{l}$ is an unbiased estimator of a node $v$ and $\theta$ is an unbiased estimator of a subgraph $G_s$.
\label{th:bias}
\end{theorem}
\begin{proof}
    Based on the property of expected value, we have:
    \begin{equation}
        E(\theta_v^{l})
        =\sum_{u\in N_v}\frac{\hat{L}_{uv}}{\alpha_{uv}}\hat{x}^{(l-1)}_u E(\mathbb{I}_{u|v}).
    \end{equation}
    Based on the property of indicator function, we have:
    \begin{equation}
    E(\theta_v^{l}) =\sum_{u\in N_v}\frac{\hat{L}_{uv}}{\alpha_{uv}}\hat{x}^{(l-1)}_u \int_{u,v}\mathbb{I}_{u|v}dP,
    \end{equation}
    where $P$ is the posterior probability of sampling the neighbour node $u$ of a given node $v$. Based on the definition of conditional probability, we have:
   \begin{equation}
    E(\theta_v^{l}) 
        =\sum_{u\in N_v}({\hat{L}_{uv}}/{\alpha_{uv}})\hat{x}^{(l-1)}_u ({p_{uv}}/{p_v}).
    \end{equation}
    Based on our definition $\alpha_{uv}={p_{uv}}/{p_v}$, we have:
    \begin{equation}         
    E(\theta_v^{l}) = \sum_{u\in N_v}{\hat{L}_{uv}}\hat{x}^{(l-1)}_u=\bar\theta_v^{l}.
    \end{equation}
    Thus, we can conclude that the estimator $\theta_v^{l}$ of node $v$ is unbiased. Next, we prove that $\theta$ is unbiased. Based on the definition of $\theta$ in Eq. \ref{eq:est_graph}, similarly, we have:
    \begin{equation}        
        E(\theta)=E(\sum_l\sum_{v\in G_s}{\theta_v^{l}}/{p_v}).
    \end{equation}
    Based on the property of a discontinuous variable, we have:
    \begin{equation}
        E(\theta)=\sum_l\sum_{v\in G_s}{\theta_v^{l}}=\bar\theta.
    \end{equation}
    Thus, the estimator $\theta$ of subgraph $G_s$ is unbiased.
\end{proof}

Sampling incurs information loss, which impacts the model quality. We minimize the variance of an estimator. By achieving the smallest possible variance, we ensure minimal deviation between the estimated value and the ``true'' value, as measured by $L_2$ norm\footnote{https://en.wikipedia.org/wiki/Efficiency\_(statistics)}. Next, we prove $\frac{n_s}{\sum_{e'}{\Vert\sum_l b^{l}_{e'} \Vert}}\Vert\sum_l b^{l}_{e} \Vert$ is the optimal sampling probability for a given edge to achieve the minimal variance of the estimator.
\begin{theorem}
Given a graph $G$ and a sampling number $n_s=\sum{p_e}$, the variance of estimator $\theta$ is minimized when $p_e=\frac{n_s}{\sum_{e'}{\Vert\sum_l b^{l}_{e'} \Vert}}\Vert\sum_l b^{l}_{e} \Vert$.
\label{th:var}
\end{theorem}
\begin{proof}
    Let $Var(\theta)$ be the variance of $\theta$. By Eq. \ref{eq:est_graph}, we have:
\begin{equation}
Var(\theta)=Var(\sum_l\sum_{e\in G_s}({b^{l}}/{p_e})\mathbb{I}_e^{l}).
\end{equation}
Based on the property of variance, we have:
\begin{gather*}
Var(\theta)=\sum_{l_1,l_2}Cov(\sum_e(\frac{b^{l_1}_e}{p_e})\mathbb{I}_e^{l_1},\sum_e(\frac{b^{l_2}_e}{p_e})\mathbb{I}_e^{l_2})\\
=\sum_l(\sum_e{Var(\frac{b^{l_1}_e}{p_e}) \mathbb{I}_e^{l})}+\sum_{e_1\neq e_2}Cov((\frac{b^{l}_{e_1}}{p_{e_1}}) \mathbb{I}_{e_1}^{l},(\frac{b^{l}_{e_2}}{p_{e_2}}) \mathbb{I}_{e_2}^{l}))\\
+ \sum_{l_1\neq l_2}(\sum_{e_1,e_2}\frac{b^{l_1}_eb^{l_2}_e}{p_{e_1}p_{e_2}}Cov(\mathbb{I}_{e_1}^{l_1},\mathbb{I}_{e_2}^{l_2}))
\end{gather*}
where $Cov(\cdot,\cdot)$ is the covariance and we have:
\begin{equation}
    Cov( \mathbb{I}_{e_1}^{l_1},\mathbb{I}_{e_2}^{l_2})=0,
\label{eq:cov0}
\end{equation}
where $e_1\neq e_2$, because each pair of different edges is sampled independently. Given an edge $e$ at different layers ($l_1\neq l_2$), the indicator functions, $\mathbb{I}_{e}^{l_1}$ and $\mathbb{I}_{e}^{l_2}$, that show if $e$ is sampled, are dependent, since they are related to the same edge. Based on the property of the indicator function, we have:
\begin{equation}
    Cov(\mathbb{I}_{e}^{l_1},\mathbb{I}_{e}^{l_2})=p_e - p_e^2.
    \label{eq:cov1}
\end{equation}
Based on Eq. \ref{eq:cov0} and \ref{eq:cov1}, we have $Var(\theta)$ as below:
\begin{gather*}
\sum_{e,l}({b^{l}_e}/{p_e})^2 
Var(\mathbb{I}_e^{l})+\sum_{e,l_1\neq l_2}({b^{l_1}_eb^{l_2}_e}/{p_e^2})Cov(\mathbb{I}_e^{l_1},\mathbb{I}_e^{l_2})\\
=\sum_e{(\sum_l b^{l}_e)^2}/{p_e} - \sum_e(\sum_l b_e^{l})^2
\end{gather*}
Based on Cauchy-Schwarz inequality, we have:
\begin{equation}
    \sum_e{((\sum_l b^{l}_e)^2}/{p_e})\sum_e p_e\geq(\sum_{e,l}b^{l}_e)^2.
\end{equation}
When $\Vert\frac{\sum_l b^{l}_e}{\sqrt{p_e}}\Vert\propto\sqrt{p_e}$ or $ p_e\propto\Vert{\sum_l b^{l}_e}\Vert$, the equality is achieved.
Thus, the variance is minimized when:
\begin{equation}
\label{eq:pe}
    p_e= \frac{n_s}{\sum_{e'} \Vert\sum_l b^{l}_{e'}\Vert}\Vert\sum_l b^{l}_{e}\Vert.
\end{equation}
\end{proof}

By Theorem \ref{th:var}, we sample a subgraph with optimal variance from edge sets at $t=1$. We extend the sampling of the subgraph at $t=1$ to overlapping sampling at later time points, and achieve optimal variance for subgraph $G_t$ by Theorem \ref{th:var2}.
\begin{theorem}
    The variance of estimator $\theta_t$ for subgraph $G_t$ is minimized when Theorem \ref{th:var} holds for both overlapping and non-overlapping samplers.
\label{th:var2}
\end{theorem} 
\begin{proof}
The sampled edges are removed from the whole edge set after these samples are selected for generating the subgraph $G_{t-1}$ at $t-1$ point. As a result, the edge set at $t$ and its previous subgraph $G_{t-1}$ do not share any common edges, thus the sampling over them is independent.
Accordingly, the variance of $\theta_t$ can be formulated as:
\begin{equation}        Var(\theta_t)=Var(\theta_O+\theta_{NO})=Var(\theta_O)+Var(\theta_{NO}),
\end{equation}
where $\theta_O$ is the estimator for the overlapping sampler over $G_{t-1}$ and $\theta_{NO}$ is the estimator for the non-overlapping sampler over the edge set at $t$.
The minimum is achieved when Theorem \ref{th:var} holds for each sampler.
\end{proof}

Now we have achieved an optimal sampling strategy. Note that the term $b^{l}_{e}$ in Eq. \ref{eq:pe} is computed based on the previous node embeddings $\hat{x}^{(l-1)}_u$ and $\hat{x}^{(l-1)}_v$, which incurs high computation cost and complex sampling process. To reduce the time cost, we approximately compute the probability by:
\begin{equation}
     \hat{p}_{e} =   \hat{L}_{uv}/(\sum_{e'} \hat{L}_{u'v'}).
\end{equation}
With this approximation, $\hat{p}_e$ ignores the previous embeddings and only depends on the graph topology. 
Since the sigmoid function is used as the activation function in GGCN, we have $\Vert \hat{x}^{l}_u\Vert\leq d$ ($d$ is the dimensionality of $\hat{x}^{l}_u$).
Thus, we can set boundaries for the sampling probability in Eq. \ref{eq:bound}:
\eat{\begin{equation}
    \frac{\hat{L}_{uv}\hat{x}_{uv}}{\sum_{e'}\hat{L}_{u'v'}d}\leq p_e\leq\frac{\hat{L}_{uv}d}{\sum_{e'}\hat{L}_{u'v'}\hat{x}_{u'v'}}.
\label{eq:bound}
\end{equation}}
\begin{equation}
\hat{L}_{uv}\hat{x}_{uv}/(\sum_{e'}\hat{L}_{u'v'}d)\leq p_e\leq\hat{L}_{uv}d/(\sum_{e'}\hat{L}_{u'v'}\hat{x}_{u'v'}),
\label{eq:bound}
\end{equation}
where $\hat{x}_{uv}=\frac{1}{2L}\Vert\sum_l (\hat{x}^{l}_u+\hat{x}^{l}_v)\Vert$. Furthermore, the error between $p_e$ and $\hat{p}_e$ can be estimated as:
\begin{equation}
\label{eq:bound_1}
     \vert\hat{p}_e - p_e \vert\leq  \vert(\hat{L}_{uv}/(\sum_{e'} \hat{L}_{u'v'}))(1-\hat{x}_{uv}/d)\vert
\end{equation}
\begin{equation}
\label{eq:bound_2}
    \vert p_e -\hat{p}_e\vert \leq \vert\hat{L}_{uv}(1/(\sum_{e'}\hat{L}_{u'v'}\hat{x}_{u'v'}/d)-1/(\sum_{e'} \hat{L}_{u'v'})) \vert
\end{equation}
In Eq. \ref{eq:bound_1}, $\hat{L}_{uv}$ and $(1-{\hat{x}_{uv}}/{d})$ are smaller than $1$. When the edge number is large, ${1}/({\sum_{e'} \hat{L}_{u'v'}})$ is small. 
For Eq. \ref{eq:bound_2}, when the edge number is large, the bound of $\vert p_e -\hat{p}_e\vert$ closes to $0$. Thus, applying GES to a large graph produces small error. 

\subsection{Dynamic Item-aware Influence Prediction}\label{sec:group_inf_pred}

Estimating the group influence and applying it to recommendation help propagate the incoming media over social networks. In practice, the media propagation probabilities among groups are affected by multiple factors like the activeness of a group, the similarity between groups, the propagation and acceptance willingness of groups for items. Meanwhile, the group influence is dynamic due to media propagation. However,  existing methods \cite{DBLP:conf/kdd/KempeKT03} initialize the propagation probability by asserting random values, which ignores all the factors on groups. The user activeness \cite{DBLP:conf/kdd/ZhangZTKLX20} and user similarity \cite{DBLP:journals/eswa/KeikhaRAA20} have been considered to compute the propagation probability. However, they all ignore the willingness of groups in media propagation. Besides, these methods statically compute the global influence and ignore the dynamics of groups.
To address these issues, we propose a DYnamic Independent Cascade (DYIC) model that considers the group activeness, group similarity, propagation and acceptance willingness with respect to items. 
First, we design a Dynamic Item-aware Information PROpagation Graph (DI$^2$PROG) to capture the dynamics of group influence with an incoming item and the temporal group-item interactions.
Then, we extend the IC model \cite{DBLP:conf/kdd/KempeKT03} by embedding DI$^2$PROG to enable its dynamics.

\noindent\textbf{DI$^2$PROG.} To predict group influence, it is necessary to discover the information propagation over groups. User groups are dynamic and may change with the incoming items over time. Thus the static information propagation graph cannot reflect the real instant group influence. A better model is required to reflect the dynamic group influence. To achieve this, we design DI$^2$PROG from which the group influence over a social network is captured to support the information propagation for selecting good user groups in recommendation. Given an incoming item $v$,
we construct a DI$^2$PROG over the embeddings of all user groups $G$ in the database and $v$. The DI$^2$PROG, $\mathcal{G}^{v}= \mathcal{(V,E')}$, is a directed graph consisting of a node set $\mathcal{V}$ and an edge set $\mathcal{E'}$. Each node is a user group. An edge from group $g_i$ to group $g_j$ denotes the probability that $g_i$ propagates item $v$ to $g_j$.
We consider three factors, the activeness $A_i$ of group $g_i$, the similarity $sim(g_i,g_j)$ between groups, the propagation willingness $W^v_i$ of group $g_i$ and that of group $g_j$, $W^v_j$, to assure the propagation probability. Following \cite{lee2021information}, the activeness is measured by the degree of recent activities against the total amount of activities.
\begin{equation}
    A_i = \frac{\#Recent~ Activities}{\# Total~Activities}.
\end{equation}
Given two groups $g_i$ and $g_j$ with their embeddings $e_{gi}$ and $e_{gj}$, we measure their similarity $sim(g_i,g_j)$ by cosine similarity between $e_{gi}$ and $e_{gj}$, and the willingness $W^v_i$ by the inner product of $e_{gi}$ and item embedding $e_v$, as well as $W^v_j$:
\begin{equation}
sim(g_i,g_j) = \frac{e_{g_i} e_{g_j}}{\Vert e_{g_i}\Vert\Vert e_{g_j}\Vert}, \quad W^v_i = e_v^Te_{g_i}.
\end{equation}
We fuse these factors to get the final edge weight $p^v_{ij}$.
\begin{equation}
    P^v_{ij} = {\gamma_1 A_i W^v_{i} + \gamma_2 sim(g_i,g_j) + (1 - \gamma_1 -\gamma_2) W^v_{j}},
\label{eq:prob}
\end{equation}
where $\gamma_1$ and $\gamma_2$ are parameters to trade-off different factors.

DI$^2$PROG is automatically updated when new group-item interactions are observed. 
Since the activeness is associated with the number of interactions, we recompute the activeness of a group $g_i$ with new interactions. The willingness of group $g_i$ is estimated to reflect the probability of item $v$ interacting $g_i$. The willingness of group $g_i$ is updated by setting $W^v_i=1$ once an interaction is observed.
We do not need to update the group similarity, since it is computed by the predicted embeddings. We periodically check the new interaction set and update the edge weights to reflect the dynamics of DI$^2$PROG.

\noindent\textbf{DYIC.}
We propose a DYIC model which simulates the information propagation over DI$^2$PROG. 
A recommended group $g_i$ is considered as the active seed that activates its neighbours based on their propagation probabilities. 
The activated neighbour groups are added into the seed set to activate the next-order neighbours. This activation process is recursively conducted until no more groups can be activated. We update DI$^2$PROG when the new interactions happen.
Alg. \ref{alg:DIC} shows the detailed process of item propagation using DYIC, which includes three steps. First, we update DI$^2$PROG by adjusting the activeness and the propagation willingness of each group in $\mathcal{V}_{new}$ (lines 2-5). Then, the groups in the current influential group seed set activate their inactive neighbours whose propagation probabilities are larger than a randomly selected threshold as in the IC model \cite{DBLP:conf/kdd/KempeKT03}. The newly activated groups or the ones contained in $\mathcal{V}_{new}$ will be added to the influential group seed set for the next-round activation. The activation is recursively conducted until no groups will be activated (lines 6-14). Finally, we merge the influential seed sets and return the number of active groups (lines 15-16). Given a DI$^2$PROG $\mathcal{G}^v$ with $N'_e$ edges, the time cost of DYIC is $\mathcal{O}(N'_e)$.

\begin{figure}[ht]\centering\small
        \vspace{0ex}
	\hspace*{-0.0cm}\fbox{
		\begin{minipage}{82mm}		\textbf{input:} the DI$^2$PROG $\mathcal{G}^v$; seed group $g_i$; streaming data $\mathcal{V}_{new}$;  \\
		~~~~~~~\textbf{output:} the number of active groups  $\vert I_g\vert$\\
			%1.for each {$g_i$ in $\mathcal{V}$} \\
			1.\hspace*{.5ex}t=1; $I^{0}_g \gets \emptyset$; $I^1_g \gets g_i$\\
			2.\hspace*{.5ex}while $t>0$ do\\
			3.\hspace*{3.5ex}for each $g_j$ in $\mathcal{V}_{new}$ ~~~~// Updating $\mathcal{G}^v$ by $\mathcal{V}_{new}$\\
			4.\hspace*{7ex}update $A_i$; $W^v_j\gets{1}$\\
			5.\hspace*{7ex}for each $g_k$ in $N_{g_j}$, \quad update $P_{jk}$ and $P_{kj}$ by Eq. \ref{eq:prob}\\
            %6.\hspace*{10ex}update $P_{jk}$ and $P_{kj}$ by Eq. \ref{eq:prob}\\
			6.\hspace*{3.5ex}for each $g_j$ in $I^{t}_g /I^{t-1}_g$\\
			7.\hspace*{7ex}for each $g_k$ in $N_{g_j}$\\
			8.\hspace*{9ex}if $P_{jk}>\tau$, \quad$g_k$ is activated\\
			9.\hspace*{9ex}if $g_k$ is active or $g_k\in \mathcal{V}_{new}$, \quad $I^{t+1} \gets g_k$\\
			10.\hspace*{2.5ex}if $I^{t+1}=\emptyset$, \quad break\\
			11.\hspace*{.5ex}$I_g = I^1_g \cup\cdots\cup I^t_g$ \\
			12.\hspace*{.5ex}Return $\vert I_g\vert$
			\end {minipage}
		} 
        \vspace{0ex}
        \caption{ DYIC model.}
        \label{alg:DIC}
        \vspace{0ex}
\end{figure}

\subsection{Group Recommendation Generation}
\label{sec:group_pref_pred}
Given an incoming item, we can simply compute the relevance scores between it and all the groups in a database, and select the top-$K$ relevant groups. In streaming, many items come in a time window, which requires relevance matching between each item and all the user groups. Suppose there are $m$ items in a time slot and $N$ groups in a database, the time cost of recommendation is $m*N$. We design a two-level hash-based User Group Index schema (UG-Index) to organize the user groups for efficient recommendation.

\begin{figure}
    \centering    \includegraphics[width=0.6\textwidth]{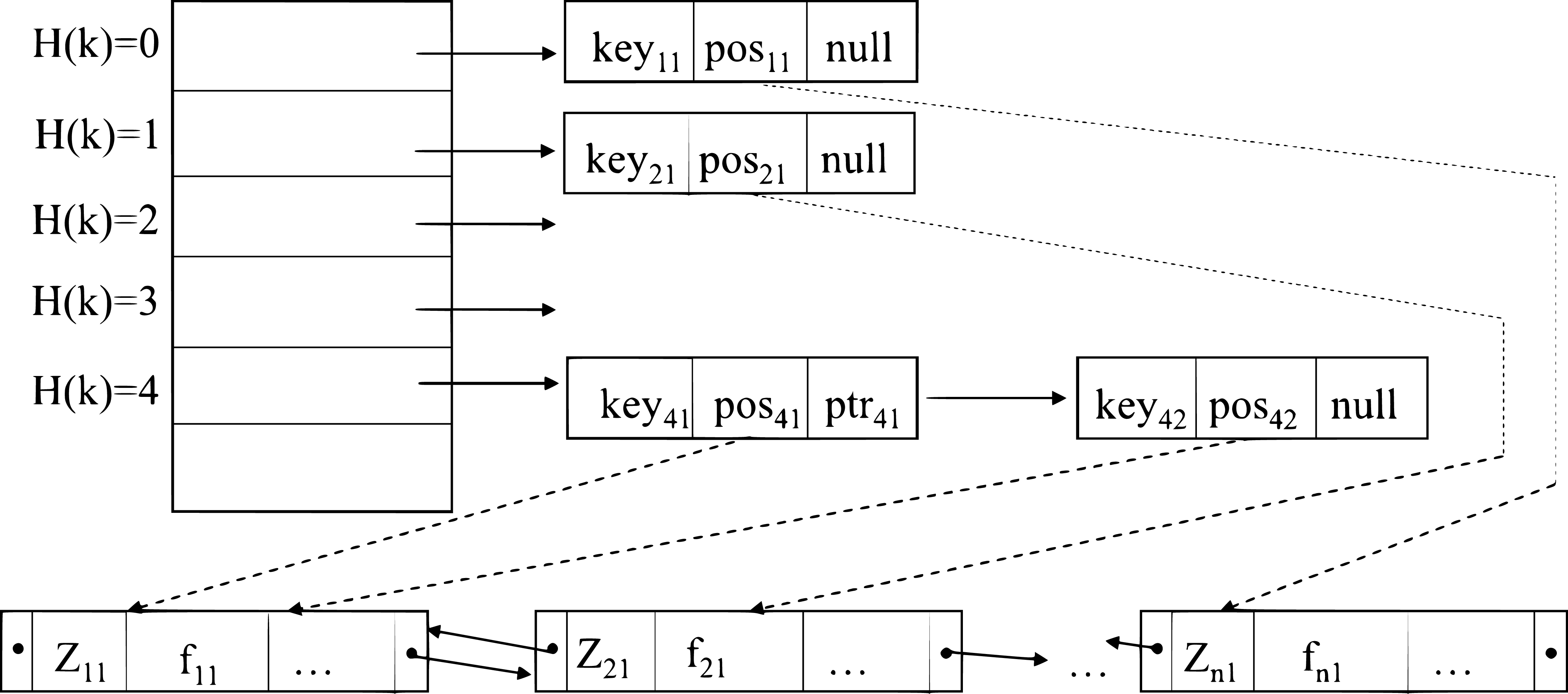}%\vspace{-3ex}
    \caption{UG-Index structure.}
    \label{fig:UG-Index}
    \vspace{0ex}
\end{figure}

As shown in Fig. \ref{fig:UG-Index}, UG-Index consists of two parts: (1) a LSH schema that maps each user group embedding to a Z-order value; and (2) a chained hash table that maps each Z-order value to the position of the corresponding Z-order value followed by its group feature. We apply LSH over group features, each is mapped into a Z-order value as in LSB-trees \cite{DBLP:conf/sigmod/TaoYSK09}. 
Inspired by the success of tensor space embedding in \cite{DBLP:journals/vldb/ZhouQ0Z19}, we embed the dot product-based relevance into a $L_2$ space so that LSH can be applied.
We map a group/item embedding to a point in a $\Bar{k}$-dimensional $L_2$ space by FastMap \cite{DBLP:conf/sigmod/FaloutsosL95} to preserve the relevance between embeddings. Given a dimension $i \in \{1 ..\Bar{k}\}$, two reference points $x_i$ and $x_i'$ are selected. Then, an embedding  is mapped into $L_2$ space by:
\begin{equation}
F_i(a)=(\hat{r}_{x_i,a}^2+\hat{r}_{x_i,x_i'}^2-\hat{r}_{x_i',a})/(2\hat{r}_{x_i,x_i'}).
\label{eq:FastMap}
\end{equation}

Given a set of group or item embeddings, we convert each group or item embedding into a $\Bar{k}$-dimensional vector that is further mapped to a Z-order value by hashing.
We apply LSH under $\hat{r}_{g,v}$ to the group-item comparison, so the approximate $k$ nearest neighbours of an item embedding can be easily found based on its Z-order value. We store these group embeddings and their Z-order values into a list of blocks in memory, where neighbouring blocks are bidirectionally connected such that the access to them can be done in two ways. Each element of these blocks is a pair of a Z-order value and the corresponding user group feature that is a vector combining the embedding of the group interaction history, group attribute feature and group influence. On top of these linked hash blocks, we construct a chained hash table to organize the positions of Z-order values in the linked blocks due to its simplicity and flexibility on its element number. Since the universal class of hash functions are unlikely to lead to a poor behaviour for a regular set of keys, we use them for mapping Z-order values to hash codes. 
Let $k$ be a Z-order value, $a$ and $b$ be randomly chosen integers, $p$ be a large prime no smaller than the total number of user groups in database, $T$ be the number of hash buckets. The hash function is defined as:
\begin{equation}
H(k)=((a\times k+b)\mod{p})\mod{T}.
\end{equation}
Given a user group set, its Z-order values is organized into a chained hash table containing a list of hash buckets. Each element of the hash table is a triplet $<key, pos, nextptr>$, where $key$ denotes the Z-order value, $pos$ is the location of the key in the LSH-based block list, and $nextptr$ is the pointer to the next element having the same hash code in the chained hash table.

\begin{algorithm}[t]
\caption{Recommendation Generation with UG-Index}
\label{alg:recGen}
\KwIn{Pretrained EIGR; item stream $I$; number of top users $K$; UG-Index}
\KwOut{$KNNlists$ --- a set of top-$K$ relevant user group lists}

$KNNlists \leftarrow \emptyset$; $\{e_i\} \leftarrow \text{ItemEmbedding}(I)$\;

\ForEach{$e_i \in \{e_i\}$}{
    GroupPreferencePrediction($I$)\;
    GroupInfluencePrediction($I$)\;
}

$\{\bar{e_i}\} \leftarrow \text{Map2L2}(\{e_i\})$; 
$\{Z_i\} \leftarrow \text{MapItem2Zorder}(\{\bar{e_i}\})$\;

$\{C_i\} \leftarrow \text{ClusterZorders}(\{Z_i\})$ \tcp*{All $Z_j$ in a $C_i$ are equal}

\ForEach{$C_i \in \{C_i\}$}{
    $pos_i \leftarrow \text{FindZorderPosition}(C_i, \text{UG-Index})$\;
    \ForEach{$c_j \in C_i$}{
        $KNNlists \leftarrow \text{MatchHashBlockElements}(c_j, pos_i)$\;
    }
}

\Return $KNNlists$\;
\end{algorithm}

Given a set of incoming items, we perform the recommendation by searching the top $K$ matched user groups in the database for each item. Alg. \ref{alg:recGen} shows the recommendation algorithm. First, we generate the item embedding for each item (line 1). Then, for each item embedding, we predict the influence and preference of user groups in the database and map it to a Z-order value by LSH-based hash mapping (lines 2-5). All the Z-order values are clustered to keep the items with the same Z-order in the same cluster (line 6). For each group, we search the chained hash table to direct to the position of its Z-order value in the  bidirectionally linked hash blocks (lines 7-8). The top \textit{K} relevant user groups for each item in a cluster are identified by computing their relevance to the candidate group at the identified Z-order position. We recursively conduct this calculation over the neighbouring candidate groups in a bidirectional order until the total number of entries accessed from all LSH blocks has reached 4B/d, where $B$ is the block size and $d$ is the dimensionality of the item or group embeddings as in \cite{DBLP:conf/sigmod/TaoYSK09} (line 10). The final recommendation for each item is returned (line 11). Compared with LSB-trees \cite{DBLP:conf/sigmod/TaoYSK09}, UG-Index does not need the tree search to find the relevant groups, which is better for stream processing.

\section{Experimental Evaluation}
\label{sec:exp}

\subsection{Experiment Setup}

We conduct experiments on three real-world datasets: Yelp\footnote{\url{https://www.yelp.com/}}, MovieLens 1M\footnote{\url{https://grouplens.org/datasets/movielens/1m/}} (ML1M), and Mafengwo\footnote{\url{https://www.mafengwo.cn/}} (MFW).
User groups are constructed following the approach in \cite{DBLP:conf/icdm/HeZ0C0T24}.
The statistics of the datasets are summarised in Table~\ref{tab:dataset}.

\begin{table}[htbp]\vspace{-0ex}
    \caption{\small Statistics of the datasets.}\vspace{-0ex}
    \centering
    \begin{tabular}{|c|c|c|c|c|c|}
    \hline
         Dataset & \# Users & \# Items & \# Groups & Avg. \# Items/Group & Avg. \# Items/User \\
    \hline
       Yelp  & 34,504 & 22,611 & 24,103 & 1.12 & 13.98 \\
    \hline
       ML1M  & 6,040 & 3,883 & 1,350 & 7.51 & 138.49 \\
    \hline
       MFW  & 5,275 & 1,513 & 995 & 3.61 & 7.53 \\
    \hline
    \end{tabular}
    \label{tab:dataset}\vspace{-0ex}
\end{table}

Following \cite{DBLP:conf/icdm/HeZ0C0T24}, all group-item interactions are chronologically ranked, with a specified time point used to split the data into the earliest 80\% as training set $D_{train}$ and the most recent 20\% as test set $D_{test}$.
The training set $D_{train}$ is further divided into $T$ subsets ${D^i_{train}}$ based on $T$ time points to train TGGCN-RA, where $i$ denotes each time interval.
The proposed EIGR model is evaluated against several state-of-the-art baselines, including LightGCN \cite{DBLP:conf/sigir/0001DWLZ020}, GroupIM \cite{DBLP:conf/sigir/SankarWWZYS20}, ConsRec \cite{DBLP:conf/www/WuX0J0ZY23}, and IGR \cite{DBLP:conf/icdm/HeZ0C0T24}.
LightGCN performs light-weight graph convolution to learn user and item embeddings. GroupIM models preference covariance among group members and assigns weights to each member to derive user and group embeddings. ConsRec leverages multi-view embeddings and a hypergraph neural network to generate group representations.
 
The effectiveness of EIGR is evaluated using standard metrics, including Hit Ratio (HR), Normalised Discounted Cumulative Gain (NDCG) \cite{DBLP:conf/sigir/Cao0MAYH18,DBLP:conf/sigir/Wang0WFC19}, and group influence \cite{DBLP:conf/icbk/QiaoM0Z21}, which is measured by the normalised number of active groups ($\sigma_{inf}$) following each recommendation.
In line with \cite{DBLP:conf/www/HeLZNHC17,DBLP:conf/sigir/WuSFHWW19}, 100 groups are randomly selected for parameter tuning, with each group treated as an individual unit in the recommendation process.
The efficiency of EIGR is assessed based on system response time. The implementation is carried out using PyTorch, and all experiments are conducted on a machine equipped with an Intel i5 2.30GHz processor, 16 GB RAM, and a 4 GB NVIDIA GTX 1050Ti GPU.
The source code and datasets are publicly available\footnote{https://github.com/MattExpCode/IGR}.

\subsection{Parameter Setting}
We evaluate the effect of parameters to get optimal results.

\begin{figure}[htb]
\centering
    \subfigure[Effect of $lr_1$]{
    \label{fig:effect_lr_1}
    \includegraphics[width=0.49\textwidth]{ 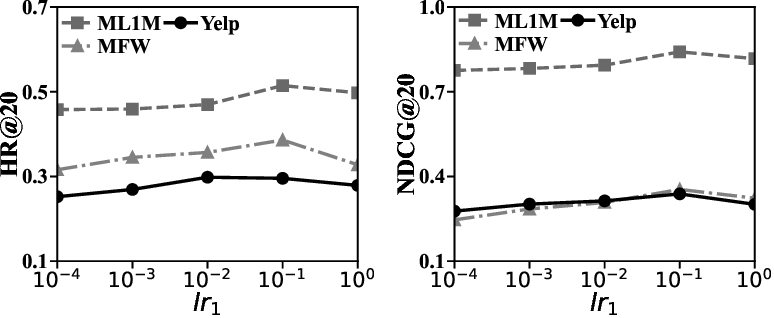}
    }%\hspace{-0.3 cm}
    \subfigure[Effect of $lr_2$]{
    \label{fig:effect_lr_2}
    \includegraphics[width=0.49\textwidth]{ 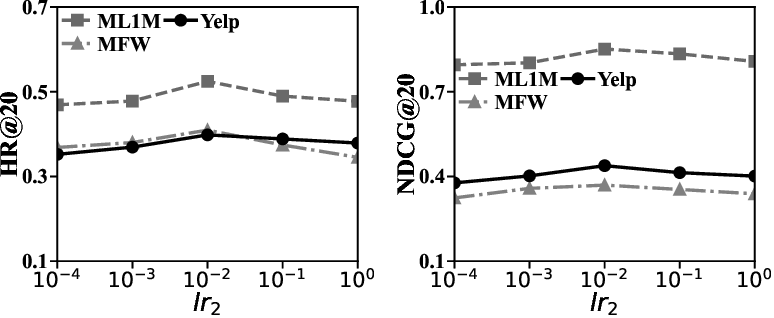}
    }%\hspace{-0.3 cm}
    \vspace{-2ex}
    \caption{Effect of $lr_i$.}
\end{figure}

\noindent\textbf{Effect of $lr_i$.}
We test the optimal learning rates $lr_i$ by varying it from $10^{-4}$ to $10^{0}$. As reported in Fig. \ref{fig:effect_lr_1} -\ref{fig:effect_lr_2}, as the increase of $lr_i$ values, the HR and NDCG values increase, reach their peaks at $lr_1 = 10^{-1}$ and $lr_2 = 10^{-2}$, and drop with the further increase of $lr_i$. Thus, we set default $lr_1$ to $10^{-1}$ and $lr_2$ to $10^{-2}$.

\begin{figure}[h]
\centering
    \subfigure[Effect of $\lambda_1$]{
    \label{fig:effect_lambda_1}
    \includegraphics[width=0.49\textwidth]{ 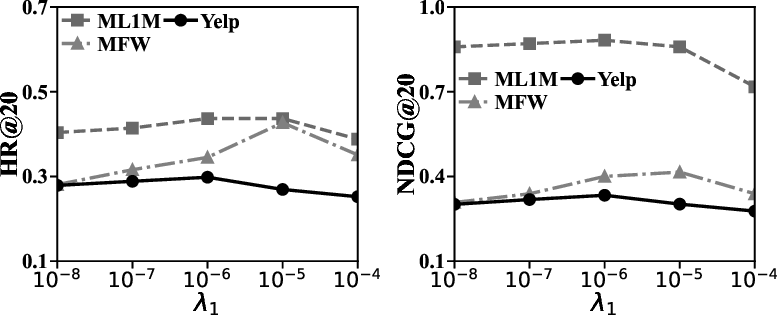}
    }%\hspace{-0.3 cm}
    \subfigure[Effect of $\lambda_2$]{
    \label{fig:effect_lambda_2}
    \includegraphics[width=0.49\textwidth]{ 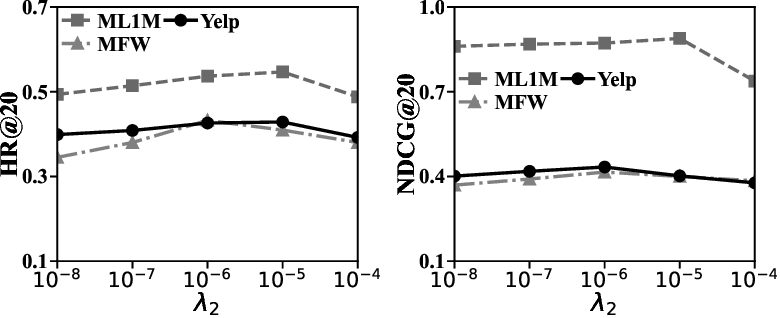}
    }%\hspace{-0.3 cm}  
    \vspace{-2ex}
    \caption{Effect of $\lambda_i$.}
\end{figure}

\noindent\textbf{Effect of $\lambda_i$.}
We evaluate the effect of regularization parameters $\lambda_1$ and $\lambda_2$ by searching them from $10^{-8}$ to $10^{-4}$. As shown in Fig. \ref{fig:effect_lambda_1} -\ref{fig:effect_lambda_2}, the HR and NDCG values increase with $\lambda_i$ increasing until they achieve optimal $\lambda_1$ and $\lambda_2$, which are $10^{-6}$ and $10^{-5}$ for Yelp and ML1M, and  $ 10^{-5}$ and $10^{-6}$ for MFW. After that, the performance drops. Thus, we set $\lambda_1 = 10^{-6}$ and $\lambda_2 = 10^{-5}$ for Yelp and ML1M, and $\lambda_1 = 10^{-5}$ and $\lambda_2 = 10^{-6}$ for MFW.

\begin{figure}[h]
    \centering
    \includegraphics[width=0.49\textwidth]{ 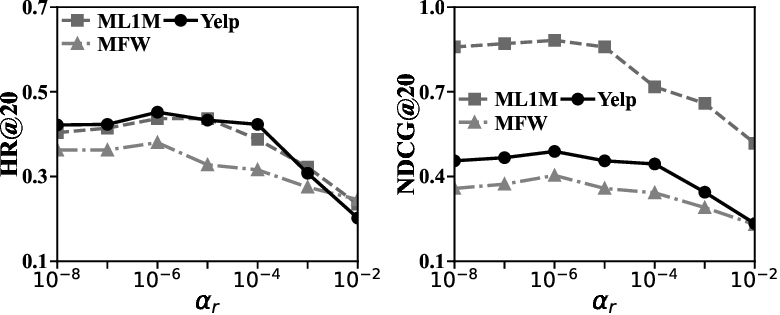}
    \vspace{-2ex}
    \caption{Effect of $\alpha_r$.}
    \label{fig:effect_alpha_r}
\end{figure}

\noindent\textbf{Effect of $\alpha_r$.}
We test the effect of the trade-off parameter $\alpha_r$ over all datasets. Fig. \ref{fig:effect_alpha_r} shows the HR and NDCG values at each $\alpha_r$. 
Clearly, the effectiveness of EIGR increases first with $\alpha_r$ decreasing, achieves the best performance at $10^{-6}$, and decreases smoothly after that. Thus, we set the default $\alpha_r$ value to $10^{-6}$. 

\begin{figure}[h]
    \centering
    \includegraphics[width=0.48\linewidth]{ 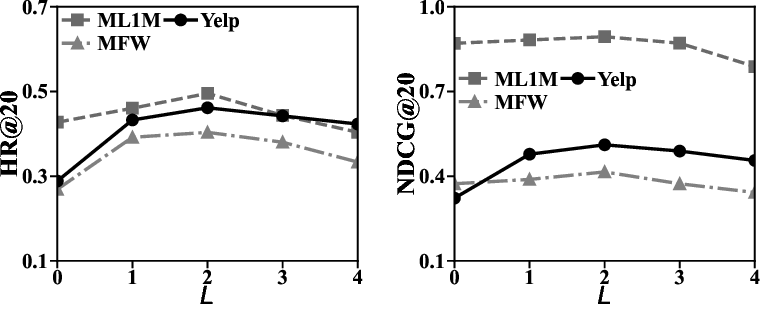}
    \vspace{-2ex}
    \caption{Effect of $L$.}
    \label{fig:effect_L}
\end{figure}

\noindent\textbf{Effect of $L$.}
We test the effect of the GCN layer number $L$ on recommendation quality by setting $L$ from 0 to 4. As shown in Fig. \ref{fig:effect_L}, the effectiveness increases with the varying of $L$ from $0$ to $2$, and degrades with the further increase of $L$ from 2 to 4. Thus, we set the default value of $L$ to $2$ for three datasets.

\begin{figure}[h]
    \centering
    \includegraphics[width=0.49\textwidth]{ 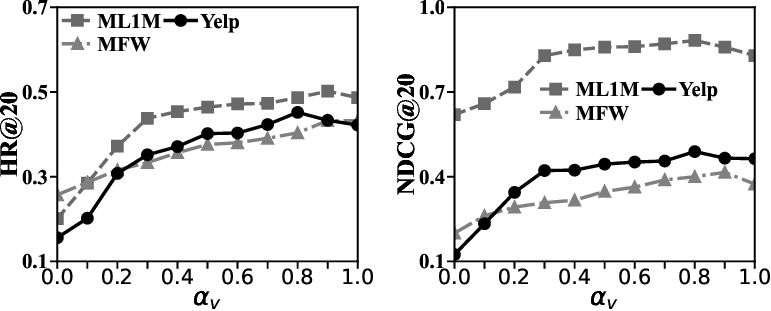}
    \vspace{-2ex}
    \caption{Effect of $\alpha_v$.}
    \label{fig:effect_alpha_v}
\end{figure}

\noindent\textbf{Effect of $\alpha_v$.}
We test the effectiveness of EIGR by varying $\alpha_v$ from 0 to 1. Fig. \ref{fig:effect_alpha_v} shows that HR curves increase with the increase of $\alpha_v$ and reach the optimal values, at $0.8$ for Yelp and $0.9$ for ML1M and MFW. 
Thus, we set the default $\alpha_v$ to $0.8$ for Yelp and $0.9$ for ML1M and MFW.

\begin{figure}[htb]
\centering
    \subfigure[Effect of $\gamma_1$]{
    \label{fig:effect_gamma_1}
    \includegraphics[width=0.49\textwidth]{ 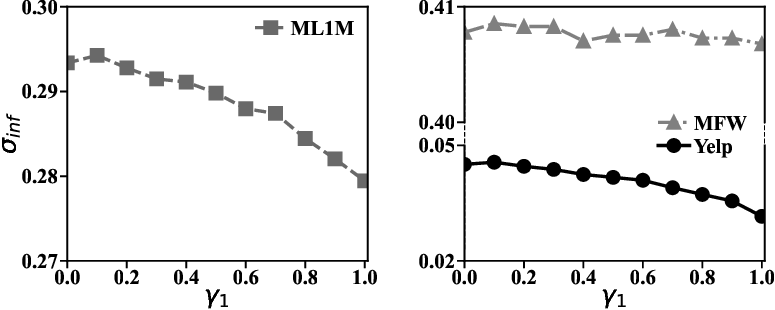}
    }%\hspace{-0.3 cm}
    \subfigure[Effect of $\gamma_2$]{
    \label{fig:effect_gamma_2}
    \includegraphics[width=0.49\textwidth]{ 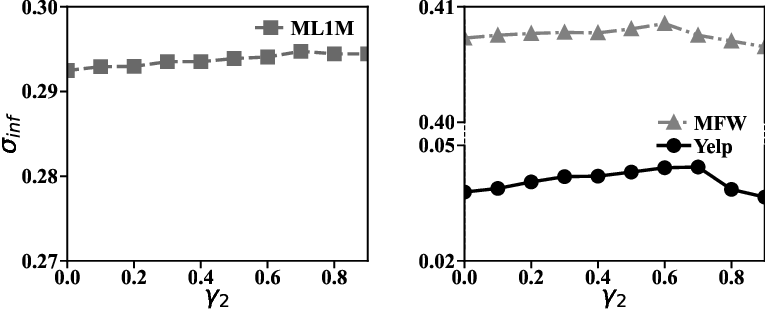}
    }%\hspace{-0.3 cm}  
    \vspace{-2ex}
    \caption{Effect of $\gamma_i$.}
\end{figure}

\noindent\textbf{Effect of $\gamma_i$.}
We evaluate the effect of $\gamma_1$ and $\gamma_2$ on the information propagation over three datasets. We first vary $\gamma_1$ from 0 to 1 and report the best $\sigma_{inf}$ value at each $\gamma_1$ for all $\gamma_2$.
As shown in Fig. \ref{fig:effect_gamma_1}, the performance increases with $\gamma_1$ increasing from 0 to 0.1, and then drops dramatically after $\gamma_1 = 0.1$.
Thus, we set the default $\gamma_1$ to $0.1$. We test the effect of $\gamma_2$ by fixing $\gamma_1$ to $0.1$ and varying $\gamma_2$ from 0 to 0.9, and reporting the HR and NDCG of recommendation results in Fig. \ref{fig:effect_gamma_2}. The $\sigma_{inf}$ increases first as the increase of $\gamma_2$, reaches the best performance at $\gamma_2 = 0.6 $ for MFW and $0.7$ for Yelp and ML1M, and degrades after the peaks. Thus, we set the default $\gamma_2 $ to $0.6$ for MFW. and $ 0.7$ for other two datasets.

\subsection{Effectiveness Evaluation}
\begin{figure}[htb]\vspace{0ex}
\centering
    \vspace{-2ex}
    \subfigure[Yelp.]{
    \label{fig:effect_cmp_yelp}
    \includegraphics[width=10.5 cm]{ 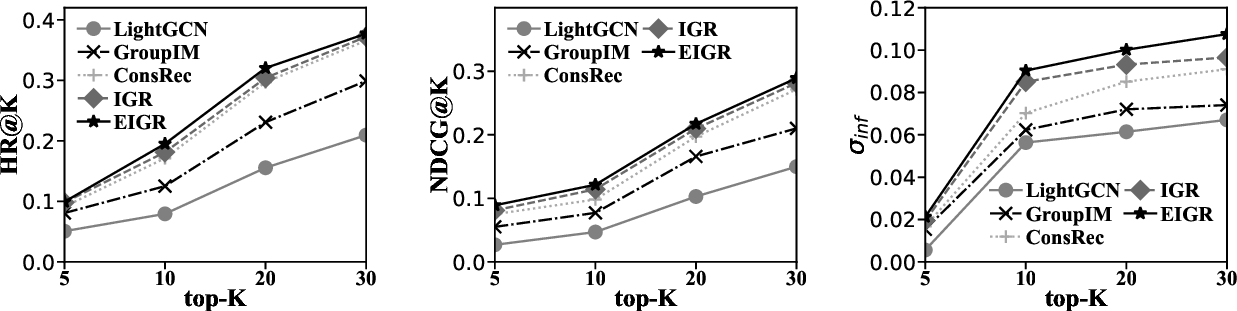}
    }
    \vspace{-2ex}
    \subfigure[ML1M.]{
    \label{fig:effect_cmp_mov}
    \includegraphics[width=10.5 cm]{ 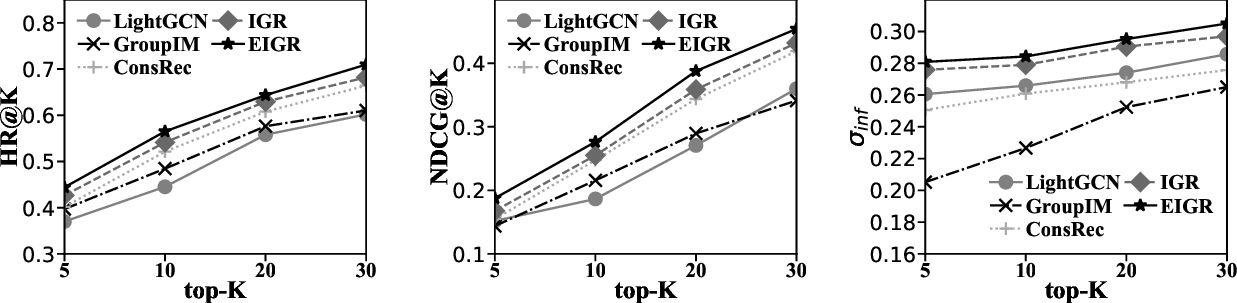}
    }
    \vspace{-2ex}
     \subfigure[MFW.]{
    \label{fig:effect_cmp_mfw}
    \includegraphics[width=10.5 cm]{ 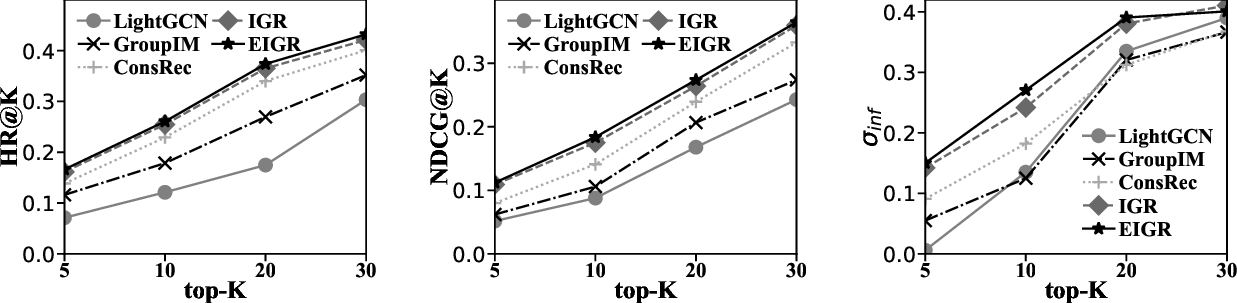}
    }
    \caption{ Effectiveness comparison.}\vspace{-0ex}
    \label{fig:effect_cmp}
\end{figure}
\noindent\textbf{Effectiveness Comparison.}
We compare EIGR with state-of-the-art baselines in terms of HR, NDCG, and $\sigma_{inf}$.
Figure~\ref{fig:effect_cmp} presents the results across the three datasets. EIGR consistently achieves the best performance, with IGR ranking second in both HR and NDCG. These results highlight the importance of modelling the dynamics of group influence for enhancing recommendation effectiveness.
For the influence measure $\sigma_{inf}$, EIGR also surpasses IGR on both Yelp and ML1M, and the two models perform comparably on MFW. These results demonstrate that dynamically modelling group influence also improves the quality of influence propagation.
Compared with other state-of-the-art methods, EIGR achieves the best performance across all datasets in terms of HR and NDCG.
This is because EIGR accounts for the dynamics of groups, items, and group influence, enabling more effective embedding learning.
ConsRec outperforms GroupIM, as it incorporates item information into group representations, thereby capturing group preferences more accurately. 
GroupIM performs worse than EIGR, IGR, and ConsRec because it relies solely on attention-based aggregation and overlooks the temporal evolution of group interests.
LightGCN yields the weakest performance across all datasets. This is likely due to its reliance on a static group-item graph structure, which becomes ineffective when group-item interactions are sparse, leading to suboptimal embedding quality.
In terms of group influence, as measured by $\sigma_{inf}$, EIGR consistently outperforms all baselines. This improvement stems from the DYIC component, which explicitly models dynamic group influence and enhances the diversity and effectiveness of influence propagation.

\subsection{Efficiency Evaluation}
\noindent\textbf{Effect of Sampling.}
We evaluate the effect of the sampling strategy on the convergence time of model training. As shown in Fig. \ref{fig:effect_sampling}, EIGR is much faster than EIGR-w/o-GES and ConsRec, and better than GroupIM. This is because EIGR reduces the redundancy across multiple graphs and generates fewer training data quadruples. Additionally, the sampling algorithm reduces the relationship between two groups from different subgraphs, leading to less training data. Thus, the designed sampling algorithm significantly improves the training efficiency while causing a slight effectiveness decline.

\begin{table}[htbp]
\centering
\caption{\small Effect of UG-Index on running time (in seconds).}
\label{tab:time_cmp}
\begin{tabular}{lccc}
\toprule
\textbf{Dataset} & \textbf{Batch-based} & \textbf{Batch+LSB-tree} & \textbf{Batch+UG-Index} \\
\midrule
Yelp   & 140.67 & 2.49 & 1.94 \\
ML1M   & 19.27  & 1.77 & 1.24 \\
MFW    & 13.60  & 1.24 & 1.01 \\
\bottomrule
\end{tabular}
\end{table}

\noindent\textbf{Effect of UG-Index.}
Given 1,000 incoming items, we test the response time of recommendation by comparing UG-Index with LSB-tree  \cite{DBLP:conf/sigmod/TaoYSK09} and the batch-based relevance matching. As shown in Table \ref{tab:time_cmp}, UG-Index performs better than LSB-tree as it exploits a two-level hashing, which identifies the positions of potential relevant groups by two hash mapping with constant time complexity. However, LSB-tree is a LSH-based B+-tree that finds the location of potential matches by one hash mapping and the Z-order value search over B+-tree. The time complexity of LSB-based search is O(log n), where $n$ is the total number of elements in the B+-tree. The batch-based approach is much slower than the other two since it browses all the groups in each batch, incurring high time cost.

\noindent\textbf{Efficiency Comparison.}
We compare the time efficiency of EIGR with state-of-the-art models across all datasets. As shown in Figure~\ref{fig:effic_cmp}, EIGR is significantly faster than both ConsRec, GroupIM, and IGR. This efficiency stems from EIGR's strategy of selecting only a small number of candidate groups for each item, thereby reducing computational overhead.

Moreover, as the dataset size increases, EIGR's time cost grows only marginally, owing to the UG-Index structure, which effectively restricts the number of candidate groups. These results demonstrate that EIGR is well-suited for efficient group recommendation on large-scale datasets.

\begin{figure}[htb]\vspace{-0ex}
\centering
    \subfigure[Effect of Sampling.]{
    \label{fig:effect_sampling}
    \includegraphics[width=0.315\linewidth]{ 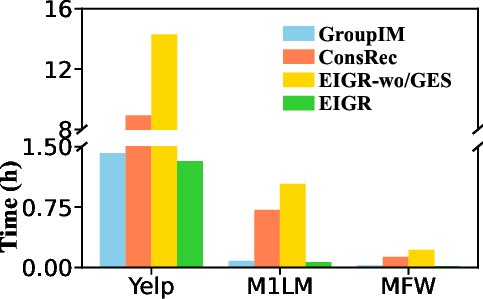}
    }
    \subfigure[Efficiency comparison.]{
    \label{fig:effic_cmp}
    \includegraphics[width=0.3\linewidth]{ 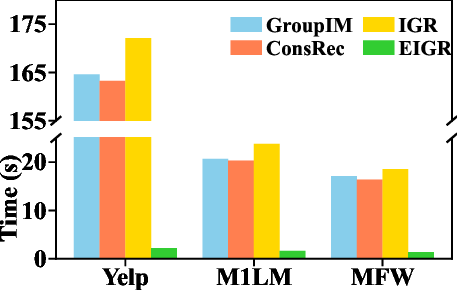}
    }
    \vspace{-0ex}
    \caption{Efficiency evaluation.}
    \label{fig:efficiency_cmp}
    \vspace{-0ex}
\end{figure}

\begin{table}[htbp]
\centering
\caption{Ablation study of EIGR on three datasets (Format: HR@20 / NDCG@20 / $\sigma_{inf}$).}
\label{tab:ablation_compact}
\renewcommand{\arraystretch}{1.1}
\setlength{\tabcolsep}{3pt}
\begin{tabular}{lccc}
\toprule
\textbf{Settings} & \textbf{Yelp} & \textbf{ML1M} & \textbf{MFW} \\
\midrule
EIGR (full)         
& 0.3221 / 0.2182 / \underline{0.1002} 
& {0.6436} / {0.3878} / {0.2953} 
& {0.3737} / {0.2736} / {0.3906} \\
w/o GES             
& \textbf{0.3244} / \textbf{0.2195} / \textbf{0.1012} 
& \textbf{0.6490} / \textbf{0.3913} / \textbf{0.2981} 
& \textbf{0.3751} / \textbf{0.2747} / \underline{0.3913} \\
w/o DYIC            
& 0.3026 / 0.2068 / 0.0940 
& 0.6251 / 0.3766 / 0.2898 
& 0.3638 / 0.2677 / 0.3783 \\
w/o UG-Index        
& \underline{0.3236} / \underline{0.2190} / 0.1009 
& \underline{0.6488} / \underline{0.3902} / \underline{0.2962} 
& \underline{0.3748} / \underline{0.2749} / \textbf{0.3919} \\
\bottomrule
\end{tabular}
\end{table}

\subsection{Ablation study}
In this subsection, we conduct the ablation study to verify the effectiveness of our key design.

\subsubsection{Ablation Analysis of EIGR}
We conduct an ablation study to assess the contributions of the three core components in EIGR: GES, DYIC, and UG-Index. The results on three datasets are reported in Table~\ref{tab:ablation_compact}.

On the Yelp dataset, removing GES yields relative increases of 0.70\% in HR@20, 0.59\% in NDCG@20, and 1.00\% in $\sigma_{inf}$. Likewise, removing UG-Index leads to increases of 0.45\%, 0.37\%, and 0.70\% respectively. These results suggest that GES and UG-Index have limited impact on recommendation effectiveness, while significantly improving model efficiency.
In contrast, removing DYIC results in consistent performance degradation. On Yelp, HR@20 drops by 6.04\% and NDCG@20 by 5.24\%, alongside a 6.18\% reduction in $\sigma_{inf}$, showing that DYIC plays a crucial role in capturing dynamic influence and maintaining both recommendation quality and influence.
Similar trends are observed on ML1M and MFW: DYIC proves essential for effectiveness, while GES and UG-Index primarily enhance efficiency with minimal impact on recommendation accuracy.

\subsubsection{Ablation Analysis of DYIC}

We conduct an ablation study to evaluate the impact of group similarity and willingness under different experimental settings. The results are presented in Table~\ref{tab:Ablation}.

\begin{table}[htb]\vspace{-0ex}
    \caption{\small Ablation analysis of DYIC model ($\sigma_{inf}$).}
    \vspace{-0ex}
    \centering
    \begin{tabular}{|c|c|c|c|}
    \hline
         Settings&  w/o Similarity & w/o Willingness & \hspace*{0.3ex} w/ All factors \hspace*{0.3ex}\\
    \hline     
     Yelp &0.0870  & 0.0901 & 0.1078 \\
   \hline
     ML1M & 0.2880 & 0.2857 & 0.3050 \\
     \hline
     MFW & 0.3979 &  0.4013 & 0.4008 \\
   \hline
    \end{tabular}
    \label{tab:Ablation}    
    \vspace{-0ex}
\end{table}

We compare the influence measure $\sigma_{inf}$ of the DYIC model under two settings: with all factors (w/ All factors) and without group similarity (w/o Similarity). The results show that incorporating all factors yields better performance. This suggests that group similarity plays an important role in enhancing the quality of influence propagation.

\noindent\textbf{Willingness.}
We compare the results generated by DYIC with all factors (w/All factors) and that without Willingness (w/o Willingness). 
Clearly, DYIC with willingness performs better, since it captures the preference of user groups with respect to items, leading to a high propagation quality. 

\section{Conclusion}
\label{sec:conclusion}
This paper studies the problem of influence-aware group recommendation over social media stream. First, we propose a novel graph sampling method to train IGR \cite{DBLP:conf/icdm/HeZ0C0T24} efficiently. Then, we propose a dynamic item-aware information propagation graph (DI$^2$PROG) to capture the group influence dynamics and a DYIC model to support the media propagation over social network dynamically. Finally, we design a UG-index to generate recommendations quickly.
The extensive experiments have proven the high efficacy of EIGR.

\begin{acks}
This work is supported by ARC Discovery Project (DP240100356), Sichuan Science and Technology Program (2025HJRC0021), and National Foreign Expert Project of China (H20240938).
\end{acks}

\bibliographystyle{ACM-Reference-Format}
\bibliography{Main}

%\appendix

\end{document}